\documentclass[12pt,draftcls,onecolumn]{IEEEtran}
\usepackage{amsmath}
\usepackage{amssymb}
\usepackage{cite}
\usepackage{color}
\usepackage{epsfig}
\usepackage{epsf}
\usepackage{rotating}
\usepackage{mathrsfs}
\usepackage{epsfig}
\usepackage{graphics}
\usepackage{theorem}

\newtheorem{thm}{Theorem}

\begin{document}

\title{Optimum Diversity-Multiplexing Tradeoff in the Multiple Relays Network \footnote{Financial supports provided by Nortel, and the corresponding matching
 funds by the Federal government: Natural Sciences and Engineering Research Council of Canada (NSERC)
 and Province of Ontario: Ontario Centres of Excellence (OCE) are gratefully acknowledged.}}

\author{\normalsize
Shahab Oveis Gharan, Alireza Bayesteh, and Amir K. Khandani \\
\small Coding \& Signal Transmission Laboratory\\[-5pt]
\small Department of Electrical \& Computer Engineering\\[-5pt]
\small University of Waterloo \\[-5pt]
\small Waterloo, ON, N2L\ 3G1 \\[-5pt]
\small {shahab, alireza, khandani}@cst.uwaterloo.ca\\}

\date{}
\maketitle \thispagestyle{empty}

\begin{abstract}
In this paper, a multiple-relay network in considered, in which
$K$ single-antenna relays assist a single-antenna transmitter to
communicate with a single-antenna receiver in a half-duplex mode.
A new Amplify and Forward (AF) scheme is proposed for this network
and is shown to achieve the optimum diversity-multiplexing
trade-off curve.
\end{abstract}
\section{System Model}
The system , as in \cite{laneman}, \cite{azarian}, and
\cite{yuksel}, consists of $K$ relays assisting the transmitter and
the receiver in the half-duplex mode, i.e. in each time, the relays
can either transmit or receive. The channels between each two node
is assumed to be quasi-static flat Rayleigh-fading, i.e. the channel
gains remain constant during a block of transmission and changes
independently from one block to another. However, we assume that
there is no direct link between the transmitter and the receiver.
This assumption is reasonable when the transmitter and the receiver
are far from each other or when the receiver is supposed to have
connection with just the relay nodes to avoid the complexity of the
network. As in \cite{azarian} and \cite{yang_belfiore}, each node is
assumed to know the state of its backward channel and, moreover, the
receiver is supposed to know the equivalent channel gain from the
transmitter to the receiver. No feedback to the transmitting node is
permitted. All nodes have the same power constraint. Also, we assume
that a capacity achieving gaussian random codebook can be generated
at each node of the network. Hence, the code design problem is not
considered in this paper.

\section{Proposed $K$-Slot Switching N-sub-block Markovian Scheme (SM)}
In the proposed scheme, the entire block of transmission is divided
into $N$ sub-blocks. Each sub-block consists of $K$ slots. Each slot
has $T'$ symbols. Hence, the entire block consists of $T=NKT'$
symbols. In order to transmit a message $w$, the transmitter selects
the corresponding codeword of a gaussian random codebook consisting
of $2^{NKT'r}$ codewords of length $\frac{NK-1}{NK}T$ and transmits
the codeword during the first $NK-1$ slots. In each sub-block, each
relay receives the signal in one of the slots and transmits the
received signal in the next slot. So, each relay is off in
$\frac{K-2}{2}$ of time. More precisely, in the $k$' slot of the
$n$'the sub-block ($1 \leq n \leq N, 1 \leq k \leq K, nk \neq NK$),
the $k$'th relay receives the signals the transmitter is sending,
and amplifies and forwards it to the receiver in the next slot. The
receiver starts receiving the signal from the second slot. After
receiving the last slot ($NK$'th slot) signal, the receiver decodes
the transmitted message by using the signal of $NK-1$ slot received
from $K$ relays. It will be shown in the next section that the
equivalent point-to-point channel from the transmitter to the
receiver would act as a lower-triangular MIMO channel.
\section{Diversity-Multiplexing Tradeoff}
In this section, we show that the proposed method achieves the
optimum achievable diversity-multiplexing curve. First, according to
the cut-set bound theorem \cite{cover_book}, the point-to-point
capacity of the uplink channel (the channel from the transmitter to
the relays) is an upper-bound for the capacity of this system.
Accordingly, the diversity-multiplexing curve of a $1 \times K$ SIMO
system which is a straight line from multiplexing gain $1$ to the
diversity gain $K$ is an upper-bound for the diversity-multiplexing
curve of our system. In this section, we prove that the tradeoff
curve of the proposed method achieves the upper-bound and thus, it
is optimum. First, we prove the statement for the case that there is
no link between the relays. Next, we prove the statement for the
general case.
\subsection{No Interfering Relays}
Assume, the link gain between the $k$'th relay and the transmitter
and the $k$'th relay and the receiver are $h_k$ and $g_k$,
respectively. Furthermore, assume that there is no link between the
relays. Accordingly, at the $k$'th relay we have
\begin{equation}
\mathbf{r}_k=h_k\mathbf{x}+\mathbf{n}_k,
\end{equation}
where $\mathbf{r}_k$ is the received signal vector of the $k$'th
relay, $\mathbf{x}$ is the transmitter signal vector and
$\mathbf{n}_k \sim \mathcal{N}(0, \mathbf{I}_{T'})$ is the noise
vector of the channel. At the receiver side, we have
\begin{equation}
\mathbf{y}=\sum_{k=1}^{K}{g_k\mathbf{t}_k}+\mathbf{z},
\end{equation}
where $\mathbf{t}_k$ is the transmitted signal vector of the $k$'th
relay, $\mathbf{y}$ is the received signal vector at the receiver
side and $\mathbf{z} \sim \mathcal{N}(0,\mathbf{I}_{T'})$ is the
noise vector of the downlink channel. The output power constraint
$\mathbb{E} \left\{\left\|\mathbf{x}\right\|^2\right\}, \mathbb{E}
\left\{\left\|\mathbf{t}_k\right\|^2\right\} \leq T'P$ holds at the
transmitter and relays side. To obtain the DM tradeoff curve of the
proposed scheme, we are looking for the end-to-end probability of
outage from the rate $r\log\left( P \right)$, as $P$ goes to
infinity.
\begin{figure}[hbt]
  \centering
  \includegraphics[scale=.5]{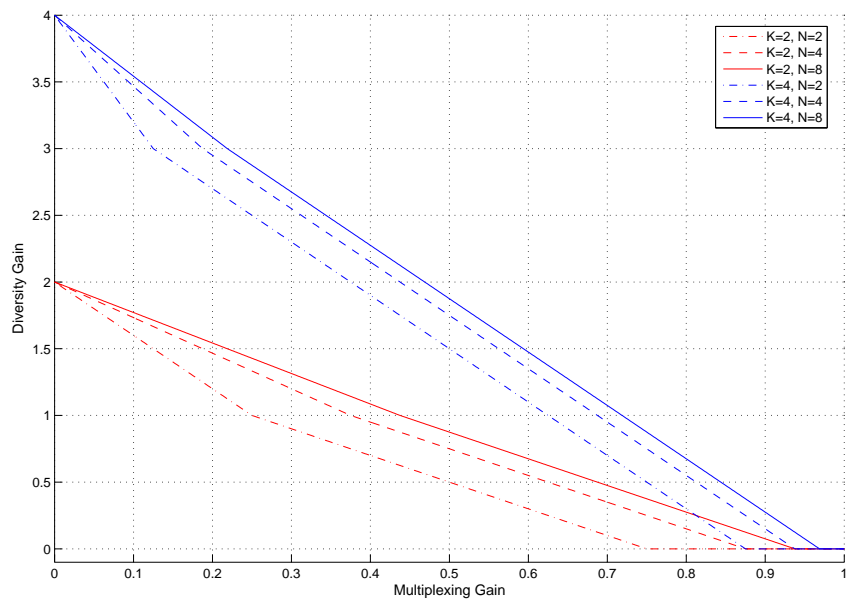}
\caption{DM Tradeoff for the proposed Switching Markovian Scheme and
various values of (K,N), No interfering relays case}
\label{fig:dm_ni}
\end{figure}
\begin{thm}
Assume a half-duplex parallel relay scenario with $K$ no interfering
relays. The proposed SM scheme achieves the diversity gain
\begin{equation}
d_{SM,NI}(r)=\max \left\{0, K\left(1-r\right)- \frac{1}{N},
K\left(1-r\right)- \frac{K r}{N-1} \right\},
\end{equation}
which achieves the optimum achievable DM tradeoff curve
$d_{opt}(r)=K(1-r)$ as $N \to \infty$.
\end{thm}
\begin{proof}
Let us define $\mathbf{x}_{n,k}, \mathbf{n}_{n,k}, \mathbf{r}_{n,k},
\mathbf{t}_{n,k}, \mathbf{z}_{n,k}, \mathbf{y}_{n,k}$ as the
signal/noise transmitted/received by the transmitter/relay/receiver
to the $k$'th relay/receiver in the $k$'th slot of the $n$'th
sub-block. Also, let us define $(k)\equiv k-2~\mod~K + 1$ and
$(n)\equiv n-\lfloor{\frac{(k)}{K}}\rfloor$. Thus, we have
\begin{eqnarray}
\mathbf{y}_{n,k}&=&g_k\mathbf{t}_{n,k}+\mathbf{z}_{n,k} \nonumber \\
&=&g_k{\alpha}_{(k)}\left(h_{(k)}\mathbf{x}_{(n),(k)}+\mathbf{n}_{(n),(k)}\right)+\mathbf{z}_{n,k},
\end{eqnarray}
where ${\alpha}_k=\frac{P}{|h_k|^2P+1}$ is the amplification
coefficient performed in the $k$'th relay. Defining the event
$\mathcal{E}_k$ as the event of outage from the rate $r\log(P)$ in
the $k$'th sub-channel consisting of the transmitter, the $k$'th
relay, and the receiver, we have
\begin{eqnarray}
\mathbb{P}\{\mathcal{E}_k\}&=&\mathbb{P}\left\{\log\left[1+P|g_k|^2|{\alpha}_k|^2|h_k|^2\left(1+|g_k|^2|{\alpha}_k|^2\right)^{-1}\right]
\leq r\log(P)\right\} \nonumber \\
&\doteq & \min \left\{ \mbox{sign}(r),
\mathbb{P}\left\{|g_k|^2|{\alpha}_k|^2|h_k|^2\left(1+|g_k|^2|{\alpha}_k|^2\right)^{-1}
\leq P^{r-1} \right\} \right\} \nonumber \\
& \stackrel{(a)}{\doteq} & \min \left\{ \mbox{sign}(r),
\mathbb{P}\left\{|g_k|^2|{\alpha}_k|^2|h_k|^2 \min
\left\{\frac{1}{2}, \frac{1}{2|g_k|^2|{\alpha}_k|^2} \right\} \leq
P^{r-1} \right\} \right\} \nonumber \\
& \stackrel{(b)}{\doteq} & \min \left\{\mbox{sign}(r), \mathbb{P}
\left\{|h_k|^2 \leq 2P^{r-1} \right\} + \mathbb{P} \left\{
|g_k|^2|{\alpha}_k|^2|h_k|^2 \leq 2
P^{r-1} \right\} \right\} \nonumber \\
& \stackrel{(c)}{\doteq} & \min \left\{ \mbox{sign}(r), P^{-(1-r)} +
\mathbb{P} \left\{ |g_k|^2\min \left\{ \frac{1}{2},
\frac{|h_k|^2P}{2} \right\} \leq 2 P^{r-1} \right\} \right\}
\nonumber \\
& \stackrel{(d)}{\doteq} & \min \left\{ \mbox{sign}(r), P^{-(1-r)} +
\mathbb{P} \left\{ |g_k|^2 \leq 4 P^{r-1} \right\} + \mathbb{P}
\left\{ |g_k|^2|h_k|^2 \leq 4 P^{r-2} \right\} \right\} \nonumber\\
& \stackrel{(e)}{\doteq} & \min \left\{ \mbox{sign}(r), P^{-(1-r)}
\right\}, \label{eq:sbch_ni}
\end{eqnarray}
where $\mbox{sign}(r)$ is the sign function, i.e. $\mbox{sign}(r)=1,
r \geq 0, \mbox{sign}(r)=0, r<0$. Here, (a) follows from the fact
that $\frac{1}{1+|g_k|^2|{\alpha}_k|^2} \doteq \min
\left\{\frac{1}{2}, \frac{1}{2|g_k|^2|{\alpha}_k|^2} \right\}$, (b)
and (d) follow from the union bound inequality, (c) follows from the
fact that $|{\alpha}_k|^2|h_k|^2 \doteq \min \left\{ \frac{1}{2},
\frac{|h_k|^2P}{2} \right\}$ and the pdf distribution of the
rayleigh-fading parameter near zero, and (e) follows from the fact
that the product of two independent rayleigh-fading parameters
behave as a rayleigh-fading parameter near zero. (\ref{eq:sbch_ni})
shows that each sub-channel's tradeoff curve performs as a
single-antenna point-to-point channel.

Defining $R_k(P)$ as the random variable showing the rate of the
$k$'th sub-channel consisting of the transmitter, the $k$'th relay,
and the receiver in terms of $P$, the outage event of the entire
channel from the $r\log(P)$, the event $\mathcal{E}$,  is equal to
\begin{equation}
\mathbb{P}\left\{\mathcal{E}\right\} = \mathbb{P}\left\{N
\sum_{k=1}^{K-1}{R_k(P)}+(N-1)R_K(P) \leq NKr\log(P) \right\}
\end{equation}
Assuming $R_k(P)= r_k\log(P)$, we have
\begin{equation}
\mathbb{P}\left\{\mathcal{E}\right\} \doteq
\mathbb{P}\left\{N\sum_{k=1}^{K-1}{r_k}+(N-1)r_K \leq NKr \right\}
\end{equation}
$\mathbb{P} \left\{R_k(P) \leq r_k \log(P)\right\}$ is known by
(\ref{eq:sbch_ni}). Defining the region $\mathcal{R}$ as
\begin{equation}
\mathcal{R} = \left\{ \left(r_1, r_2, \cdots, r_K\right) | 0 \leq
r_k \leq 1,  N\sum_{k=1}^{K-1}{r_k}+(N-1)r_K \leq NKr \right\}
\label{eq:R_df_ni}
\end{equation}
it is easy to check that all the vectors $\left(r_1, r_2, \cdots,
r_K\right)$ that result in the outage event almost surely lie in
$\mathcal{R}$. In fact, according to (\ref{eq:sbch_ni}), for all $k$
we know $r_k \geq 0$. Also, for $r_k
> 1$, $\mathbb{P} \left\{R_k(P) \geq r_k \log(P)\right\} \leq
e^{-P^{r-1}}$ which is exponential in terms of $P$. Hence, $r_k > 1$
can be disregarded for the outage region. As a result, $\mathbb{P}
\left\{ \mathcal{E} \right\} \doteq \mathbb{P} \left\{ \mathbf{r}
\in \mathcal{R} \right\}$.

On the other hand, by (\ref{eq:sbch_ni}) and the fact that $r_k$'s
are independent, we have
\begin{equation}
\mathbb{P} \left\{ r_1 \leq r_1^0, r_2 \leq r_2^0, \cdots, r_K \leq
r_K^0 \right\} \doteq P^{-\left(K-\sum_{k=1}^{K}{r_k^0}\right)}
\label{eq:cdf_ni}
\end{equation}
Now, we show that $\mathbb{P} \left\{ \mathcal{E} \right\} \doteq
P^{-\min_{\mathbf{r} \in
\mathcal{R}}{K-\mathbf{1}\cdot\mathbf{r}}}$. First of all, by taking
derivative of (\ref{eq:cdf_ni}) with respect to
$r_1,r_2,\cdots,r_K$, it is easy to see that the probability density
function of $\mathbf{r}$ behaves the same as the probability
function in (\ref{eq:cdf_ni}), i.e. $f_r(\mathbf{r}) \doteq
P^{-\left( {K-\mathbf{1}\cdot\mathbf{r}} \right) }$. Hence, the
outage probability is equal to
\begin{eqnarray}
\mathbb{P} \left\{ \mathcal{E} \right\} & \doteq & \int_{\mathbf{r}
\in \mathcal{R}}{f_r(\mathbf{\mathbf{r}})d\mathbf{r}} \nonumber \\
& \dot{\leq} & vol(\mathcal R)P^{-\min_{\mathbf{r} \in
\mathcal{R}}{K-\mathbf{1}\cdot\mathbf{r}}} \nonumber \\
& \stackrel{(a)}{\doteq} & P^{-\min_{\mathbf{r} \in
\mathcal{R}}{K-\mathbf{1}\cdot\mathbf{r}}} \label{eq:R_ub_ni}
\end{eqnarray}
Here, (a) follows from the fact that $\mathcal{R}$ is a fixed
bounded region whose volume is independent of $P$. On the other
hand, by continuity of $P^{-\left( {K-\mathbf{1}\cdot\mathbf{r}}
\right) }$ over $\mathbf{r}$, we have $\mathbb{P} \left\{
\mathcal{E} \right\} \dot{\geq} P^{-\min_{\mathbf{r} \in
\mathcal{R}}{K-\mathbf{1}\cdot\mathbf{r}}}$ which combining with
(\ref{eq:R_ub_ni}), results into $\mathbb{P} \left\{ \mathcal{E}
\right\} \doteq P^{-\min_{\mathbf{r} \in
\mathcal{R}}{K-\mathbf{1}\cdot\mathbf{r}}}$. Defining
$l(\mathbf{r})=K-\mathbf{1}\cdot\mathbf{r}$, we have to solve the
following linear programming optimization problem $\min_{\mathbf{r}
\in \mathcal{R}}{l(\mathbf{r})}$. Notice that the region
$\mathcal{R}$ is defined by a set of linear inequality constraints.
To solve the problem, we have
\begin{eqnarray}
l(\mathbf{r}) & \stackrel{(a)}{\geq} & \max \left\{0, K - \frac{NKr + r_K}{N}, K - \frac{NKr-\sum_{k=1}^{K-1}r_k}{N-1} \right\}\nonumber \\
& \stackrel{(b)}{\geq} & \max \left\{0, K(1-r)- \frac{1}{N}, K(1-r)-
\frac{K r}{N-1} \right\}.
\end{eqnarray}
Here, (a) follows from the inequality constraint in
(\ref{eq:R_df_ni}) governing $\mathcal{R}$, and (b) follows from the
fact that $r_K \leq 1$ and $\forall k<K: r_k \geq 0$. Now, we
partition the range $0 \leq r \leq 1$ into three intervals. First,
in the case that $r>1-\frac{1}{NK}$, the feasible point
$\mathbf{r}=\mathbf{1}$ achieves the lower bound $0$. Second, in the
case that $r<\frac{1}{K}-\frac{1}{NK}$, the feasible point
$\mathbf{r}=\left(0,0,\cdots,0,\frac{NKr}{N-1}\right)$, achieves the
lower bound $K(1-r)- \frac{K r}{N-1}$. Finally, in the case that
$\frac{1}{K}-\frac{1}{NK} \leq r \leq 1-\frac{1}{NK}$, The lower
bound $K(1-r)- \frac{1}{N}$ is achievable by the feasible point
$\mathbf{r}, \forall k<K:~ r_k=\frac{NKr-N+1}{N(K-1)}, r_K=1$.
Hence, we have $\min_{\mathbf{r} \in \mathcal{R}} l(\mathbf{r}) =
\max \left\{ 0, K(1-r)- \frac{1}{N}, K(1-r)- \frac{K r}{N-1}
\right\}$. This completes the proof.
\end{proof}
\textit{Remark -} It is worth noting that as long as the graph $G(V,
E)$ whose vertices are the relay nodes and edges are the non
interfering relay node pairs includes a hamiltonian cycle
\footnote{By hamiltonian cycle, we mean a simple cycle $v_1v_2\cdots
v_K v_1$ that goes exactly one time through each vertex of the
graph.}, the result of this subsection remains valid.
\subsection{General Case}
In the general case, an interference term due to the neighboring
relay adds at the receiver antenna of each relay.
\begin{equation}
\mathbf{r}_k = h_k \mathbf{x} + i_{(k)} \mathbf{t}_{(k)} +
\mathbf{n}_k,
\end{equation}
where $i_{(k)}$ is the interference link gain between the $k$'th and
$(k)$'th relays. Hence, the amplification coefficient is bounded as
$\alpha _k \leq \frac{P}{P \left( \left| h_k \right|^2 + \left|
i_{(k)} \right|^2 \right) + 1}$. Here, we observe that in the case
that $\alpha_k >1$, the noise $n_k$ at the receiving side of the
$k$'th relay can be boosted at the receiving side of the next relay.
Hence, we bound the amplification coefficient as $\alpha_k = \min
\left\{ 1, \frac{P}{P \left( \left| h_k \right|^2 + \left| i_{(k)}
\right|^2 \right) + 1} \right\}$. In this way, it is guaranteed that
the noise of relays are not boosted up through the system. This is
at the expense working with the output power less than $P$. On the
other hand, we know that almost surely \footnote{By almost surely,
we mean its probability is greater than $1-P^{-\delta}$, for all
values of $\delta$.} $\left| h_k \right|^2 , \left| i_{(k)}
\right|^2 \dot{\leq} 1$. Hence, almost surely we have $\alpha _k
\doteq 1$. Another change we make in this part is that we assume
that the entire time of transmission consists of $NK+1$ slots, and
the transmitter sends the data during the first $NK$ slots while the
relays send in the last $NK$ slots (from the second slot up to the
$NK+1$'th slot). Hence, we have $T=(NK+1)T'$. This assumption makes
our analysis easier and the lower bound on the diversity curve
tighter. Now, we prove the main theorem of this section.
\begin{thm}
Consider a half-duplex multiple relays scenario with $K$ interfering
relays whose gains are independent rayleigh fading variables. The
proposed SM scheme achieves the diversity gain
\begin{equation}
d_{SM,I}(r) \geq \max \left\{ 0, K \left( 1 - r \right) -
\frac{r}{N} \right\},
\end{equation}
which achieves the optimum achievable DM tradeoff curve
$d_{opt}(r)=K(1-r)$ as $N \to \infty$.
\end{thm}
\begin{proof}
First, we show that the entire channel matrix acts as a lower
triangular matrix. At the receiver side, we have
\begin{eqnarray}
\mathbf{y}_{n,k} & = & g_k \mathbf{t}_{n, k} + \mathbf{z}_{n, k}
\nonumber \\
& = & g_k \alpha _{(k)} \left( \sum_{0 < n_1, k_1, n_1 (K + 1) + k_1
< n (K + 1) + k}{p_{n-n_1, k, k_1}\left( h_{k_1}\mathbf{x}_{n_1,
k_1} + \mathbf{n}_{n_1, k_1}\right) }  \right) + \mathbf{z}_{n, k}
\end{eqnarray}
Here, $p_{n, k, k_1}$ has the following recursive formula $p_{0, k,
k}=1, p_{n, k, k_1}=i_{(k)}\alpha_{(k)}p_{(n), (k), k_1}$. Defining
the square $NK \times NK$ matrices as $\mathbf{G}= \mathbf{I}_N
\otimes \textit{diag}\left\{ g_1, g_2, \cdots ,g_K \right\}$,
$\mathbf{H}= \mathbf{I}_N \otimes \textit{diag}\left\{ h_1, h_2,
\cdots ,h_K \right\}$, $\mathbf{\Omega} = \mathbf{I}_N \otimes
\textit{diag}\left\{ \alpha _1, \alpha _2, \cdots ,\alpha _K
\right\}$, and
\begin{equation}
\mathbf{F}= \left(
\begin{array}{ccccc}
1 & 0 & 0 & 0 & \ldots \\
p_{0,2,1} & 1 & 0 & 0 & \ldots \\
p_{0,3,1} & p_{0, 3, 2} & 1 & 0 & \ldots \\
\vdots & \vdots & \vdots & \vdots & \ddots \\
p_{N-1, K, 1} & p_{N-1, K, 2} & \ldots & p_{0, K, K-1} & 1
\end{array},
 \right)
\end{equation}
where $\otimes$ is the Kronecker product\cite{matrix_book} of
matrices and $\mathbf{I}_N$ is the $N \times N$ identity matrix, and
the $NK \times 1$ vectors $\mathbf{x}\left(s\right)=[x_{1,1}(s),
x_{1, 2}(s), \cdots ,x_{N, K}(s)]^T$,
$\mathbf{n}\left(s\right)=\left[n_{1,1}\left(s\right), n_{1, 2}(s),
\cdots ,n_{N, K}(s)\right]^T$,
$\mathbf{z}\left(s\right)=[z_{1,2}(s), z_{1, 3}(s), \cdots ,z_{N+1,
1}(s)]^T$, and $\mathbf{y}\left(s\right)=[y_{1,2}(s), y_{1, 3}(s),
\cdots ,y_{N+1, 1}(s)]^T$, we have
\begin{equation}
\mathbf{y}\left(s\right) = \mathbf{G} \mathbf{\Omega} \mathbf{F}
\left( \mathbf{H} \mathbf{x}\left(s\right) +
\mathbf{n}\left(s\right) \right) + \mathbf{z}\left(s\right).
\end{equation}
Here, we observe that the matrix of the entire channel acts as a
lower triangular matrix of a $NK \times NK$ MIMO channel whose noise
is colored. The probability of outage of such a channel for the
multiplexing gain $r$ is defined as
\begin{equation}
\mathbb{P} \left\{ \mathcal{E} \right\}=\mathbb{P} \left\{ \log
 \left|\mathbf{I}_{KN} + P \mathbf{H}_{T}\mathbf{H}_{T}^{H}\mathbf{P}_n^{-1} \right| \leq
(NK+1)r \log\left( P \right) \right\},
\end{equation}
where $\mathbf{P}_N=\mathbf{I}_{NK}+\mathbf{G} \mathbf{\Omega}
\mathbf{F} \mathbf{F}^H \mathbf{\Omega}^H \mathbf{G}^H$, and
$\mathbf{H}_T=\mathbf{G} \mathbf{\Omega} \mathbf{F} \mathbf{H}$.
Assume $|h(k)|^2=P^{-\mu(k)}$, $|g(k)|^2=P^{-\nu(k)}$,
$|i(k)|^2=P^{-\omega(k)}$, and $\mathcal{R}$ as the region in
$\mathbb{R}^{3K}$ that defines the outage event $\mathcal{E}$ in
terms of the vector $[\mathbf \mu, \mathbf \nu, \mathbf \omega]$,
where $\mathbf{\mu}=\left[ \mu(1) \mu(2) \cdots \mu(K) \right]^T,
\mathbf{\nu}=\left[ \nu(1) \nu(2) \cdots \nu(K)
\right]^T,\mathbf{\omega}=\left[ \omega(1) \omega(2) \cdots
\omega(K) \right]^T$. The probability distribution function (and
also the inverse of cumulative distribution function) decays
exponentially as $P^{-P^{-\delta}}$ for positive values of $\delta$.
Hence, the outage region $\mathcal R$ is almost surely equal to
$\mathcal{R}_{+}=\mathcal{R} \bigcap \mathbb{R}_{+}^{3K}$. Now, we
have
\begin{eqnarray}
\mathbb{P} \left\{ \mathcal{E} \right\} & \stackrel{(a)}{\leq} &
\mathbb{P} \left\{ \left| \mathbf{H}_T \right|^2 \left| \mathbf{P}_n
\right|^{-1} \leq
P^{-NK \left( 1-r \right) +r}\right\} \nonumber \\
& \stackrel{(b)}{\leq} & \mathbb{P} \left\{ -N
\sum_{k=1}^{K}{\mu(k)+\nu(k)- \min \left\{ 0, \mu(k), \omega((k))
\right\}} + \right. \nonumber \\
&& \frac{NK\log(3) - \log \left| \mathbf{P}_{N} \right|}{\log \left(
P \right)} \leq -NK(1-r)+r
\Bigg\} \nonumber \\
& \stackrel{(c)}{\dot{\leq}} & \mathbb{P} \left\{ NK \frac{\log(3) -
\log (N^2K^2+1)}{\log (P)} + NK\left( 1-r \right) - r \leq N
\sum_{k=1}^{K}{\mu(k) + \nu(k)},\right. \nonumber \\
&&  \mu(k),\nu(k),\omega(k) \geq 0 \Bigg\}. \label{eq:R_hat_wi}
\end{eqnarray}
Here, (a) follows from the fact that for a positive semidefinite
matrix $\mathbf A$ we have $\left| \mathbf{I} + \mathbf{A} \right|
\geq \left| \mathbf{A} \right|$, (b) follows from the fact
that
\begin{equation}
\alpha (k) = \min \left\{ 1, \frac{P}{P^{1-\mu (k)} + P^{1-\omega
((k))} + 1} \right\} \geq \frac{1}{3} \min \left\{ 1,  P, P^{\mu
(k)}, P^{\omega ((k))} \right\}  \nonumber
\end{equation}
and assuming $P$ is large enough such that $P \geq 1$, and (c)
follows from the fact that $\alpha (k) \leq 1$ and accordingly,
$p_{n,k,k_1} \leq 1$, and knowing that the sum of the entries of
each row in $\mathbf{F}\mathbf{F}^H$ is less than $N^2K^2$, we
have\footnote{This can be verified by the fact that every symmetric
real matrix $\mathbf{A}$ which has the property that for every $i$,
$a_{i,i} \geq \sum_{i \neq j}|a_{i,j}|$ is positive semidefinite.}
$\mathbf{F}\mathbf{F}^H \preccurlyeq N^2K^2 \mathbf{I}_{NK}$, and
$\mathbb P \left\{ \mathcal{R} \right\} \doteq \mathbb P \left\{
\mathcal{R}_{+} \right\}$, and conditioned on $\mathcal{R}_{+}$, we
have $\min \left\{ 0, \mu(k), \omega ((k)) \right\} = 0$ and $\nu
(k) \geq 0$ and consecutively $\mathbf{P}_N \preccurlyeq (N^2K^2 +
1) \mathbf{I}_{KN}$.

On the other hand, we know for vectors $\mathbf{\mu}^0,
\mathbf{\nu}^0, \mathbf{\omega}^0 \geq \mathbf 0$, we have
$\mathbb{P} \left\{\mathbf{\mu} \geq \mathbf{\mu}^0, \mathbf{\nu}
\geq \mathbf{\nu}^0, \mathbf{\omega} \geq \mathbf{\omega}^0 \right\}
\doteq P^{-\mathbf{1} \cdot \left( \mathbf{\mu}^0 + \mathbf{\nu}^0 +
\mathbf{\omega}^0 \right)}$. Similarly to the proof of Theorem 1, by
taking derivative with respect to $\mathbf \mu, \mathbf \nu$ we have
$f_{\mathbf \mu, \mathbf \nu}(\mathbf \mu, \mathbf \nu) \doteq P^{-
\mathbf 1 \cdot \left( \mathbf \mu + \mathbf \nu \right)}$ .Defining
the lower bound $l_0$ as $l_0 = \frac{\log(3) - \log
(N^2K^2+1)}{\log (P)} + \left( 1-r \right) - \frac{r}{NK} $, the new
region $\hat{\mathcal{R}}$ as $\hat{\mathcal{R}} = \left\{
\mathbf{\mu},\mathbf{\nu} \geq \mathbf{0}, \frac{1}{K} \mathbf{1}
\cdot \left( \mathbf{\mu} + \mathbf{\nu}\right) \geq l_0 \right\}$,
the cube $\mathcal I$ as $\mathcal I = \left[0, Kl_0 \right]^{2K}$,
and for $1 \leq i \leq 2K$, $\mathcal{I}_i^c=[0, \infty )^{i-1}
\times [Kl_0, \infty ) \times [0, \infty )^{2K-i}$,
 we observe
\begin{eqnarray}
\mathbb{P} \left\{ \mathcal E \right\} & \stackrel{(a)}{\dot{\leq}}
& \mathbb{P} \{ \hat{\mathcal R} \} \nonumber \\
& \leq & \int_{\mathcal{\hat{R}} \bigcap \mathcal{I}}{f_{\mathbf
\mu, \mathbf \nu}\left(\mathbf \mu, \mathbf \nu \right) d\mathbf \mu
d \mathbf \nu} + \sum_{i=1}^{2K}{\mathbb{P} \left\{ [\mathbf \mu,
\mathbf \nu] \in  \mathcal{\hat{R}} \cap \mathcal{I}_i^c \right\}}
\nonumber
\\
&\dot{\leq} & vol (\mathcal{\hat{R}} \cap \mathcal{I})
P^{-\min_{\left[ \mathbf{\mu}^0, \mathbf{\nu}^0 \right] \in
\mathcal{\hat{R}} \bigcap \mathcal{I}} \mathbf{1} \cdot \left(
\mathbf{\mu}^0 + \mathbf{\nu}^0 \right) } + 2K P^{-Kl_0}
\nonumber \\
& \stackrel{(b)}{\doteq} & P^{-Kl_0} \nonumber \\
& \doteq & P^{-\left[K \left( 1 - r \right) - \frac{r}{N} \right]}.
\label{eq:t2_r_wi}
\end{eqnarray}
Here, (a) follows from (\ref{eq:R_hat_wi}) and (b) follows from the
fact that $\mathcal{\hat{R}} \bigcap \mathcal{I}$ is a bounded
region whose volume is independent of $P$. (\ref{eq:t2_r_wi})
completes the proof.
\end{proof}
\begin{figure}[hbt]
  \centering
  \includegraphics[scale=.5]{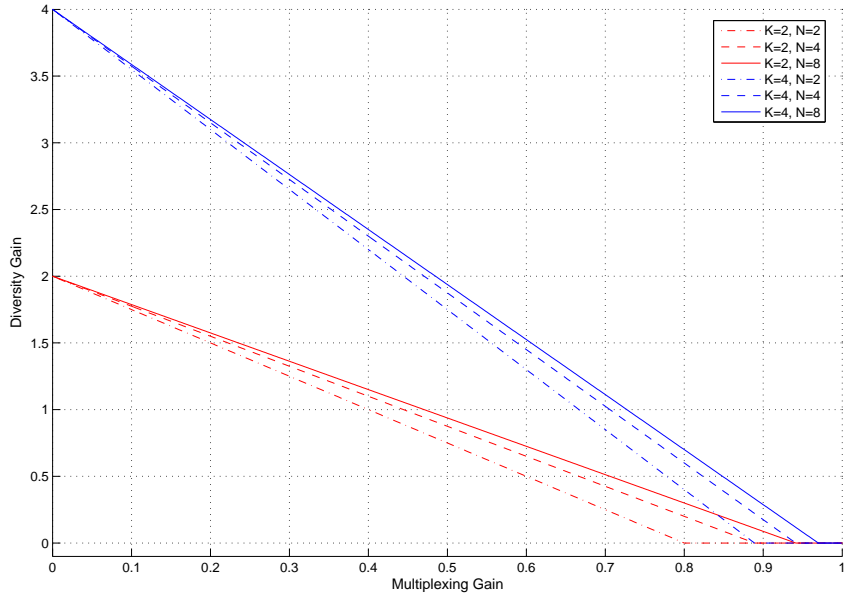}
\caption{DM Tradeoff for the proposed Switching Markovian Scheme and
various values of (K,N), Interfering relays case} \label{fig:dm_wi}
\end{figure}
\textit{Remark -} The statement in the above theorem holds for the
general case in which any arbitrary set of relay pairs are
non-interfering. Hence, the proposed scheme achieves the upper-bound
of the tradeoff curve in the asymptotic case of $N \to \infty$ for
any graph topology on the interfering relay pairs.

Figure (\ref{fig:dm_wi}) shows the D-M tradeoff curve of the scheme
for the case of interfering relays and varying number of $K$ and
$N$.

\end{document}